\documentclass{my-aims}

\usepackage{txfonts}
\def\typeofarticle{Research article}

\def\ppages{}
\def\DOI{10.3934/Math.2018.x.xxx}
\def\Received{July 2018}
\def\Accepted{October 2018}
\def\Published{xxx}

\numberwithin{equation}{section}

\makeatletter
\renewcommand{\@biblabel}[1]{#1\hfill \hspace{-0.2cm}}
\makeatother

\theoremstyle{plain}
\newtheorem{theorem}{Theorem}
\newtheorem{proposition}[theorem]{Proposition}
\newtheorem{lemma}[theorem]{Lemma}

\theoremstyle{remark}
\newtheorem{remark}{Remark}

\theoremstyle{definition}
\newtheorem{definition}{Definition}

\begin{document}

\title{A cholera mathematical model with vaccination
and the biggest outbreak\\ of world's history}


\author{%
Ana P. Lemos-Pai\~{a}o,
Cristiana J. Silva 
and
Delfim F. M. Torres\corrauth
}

\shortauthors{the Author(s)}


\address{%
\addr{}{Center for Research and Development in Mathematics and Applications (CIDMA),\\
Department of Mathematics, University of Aveiro, 3810-193 Aveiro, Portugal}}

\corraddr{Email: delfim@ua.pt; Tel: +351 234370668; Fax: +351 234370066.}


\begin{abstract}
We propose and analyse a mathematical model for cholera considering vaccination. 
We show that the model is epidemiologically and mathematically well posed and 
prove the existence and uniqueness of disease-free and endemic equilibrium points. 
The basic reproduction number is determined and the local asymptotic stability 
of equilibria is studied. The biggest cholera outbreak of world's history 
began on 27$^{\text{th}}$ April 2017, in Yemen. Between 27$^{\text{th}}$ 
April 2017 and 15$^{\text{th}}$ April 2018 there were 2275 deaths 
due to this epidemic. A vaccination campaign began on 6$^{\text{th}}$ 
May 2018 and ended on 15$^{\text{th}}$ May 2018. 
We show that our model is able to describe well this outbreak. 
Moreover, we prove that the number of infected individuals would have been 
much lower provided the vaccination campaign had begun earlier.
\end{abstract}


\maketitle


\section{Introduction}

Cholera is an acute diarrhoeal illness caused by infection of
the intestine with bacterium \textit{vibrio cholerae}, 
which lives in an aquatic organism \cite{Lemos-Paiao}. 
The ingestion of contaminated water
can cause cholera outbreaks, as John Snow proved in 1854 \cite{Shuai}.
This is a possibility of transmission of the disease, but others do exist. 
For example, susceptible individuals can become infected 
if they become in contact with contaminated people. When individuals 
are at an increased risk of infection, then they can transmit 
the disease to those who live with them, by reflecting 
food preparation or using water storage containers \cite{Shuai}. 
An individual can be infected with or without symptoms and some of these 
can be watery diarrhea, vomiting and leg cramps. 
If contaminated individuals do not get treatment, 
then they become dehydrated, suffering acidosis and circulatory
collapse. This situation can lead to death within 12 to 24 hours \cite{Mwasa,Shuai}.
Some studies and experiments suggest that a recovered individual can be immune
to the disease during a period of 3 to 10 years. Nevertheless, recent researches 
conclude that immunity can be lost after a period of weeks to months \cite{Neilan,Shuai}.

Since 1979, several mathematical models for the 
transmission of cholera have been proposed: see, e.g.,
\cite{Capasso:1979,Capone,Codeco:2001,Hartley:2006,Hove-Musekwa,Joh:2009,%
Lemos-Paiao, Mukandavire:2008,Mwasa,Neilan,Pascual,Shuai} 
and references cited therein. In \cite{Neilan}, 
the authors propose a SIR (Susceptible--Infectious--Recovered)
type model and consider two classes of bacterial concentrations 
(hyperinfectious and less-infectious) and two classes 
of infectious individuals (asymptomatic and symptomatic). 
In \cite{Shuai}, another SIR-type model is proposed 
that incorporates hyperinfectivity (where infectivity 
varies with time since the pathogen was shed)
and temporary immunity, using distributed delays. 
The authors of \cite{Mwasa} incorporate in a SIR-type model 
public health educational campaigns, vaccination, quarantine and treatment, 
as control strategies in order to curtail the disease. 

Between 2007 and 2018, several cholera outbreaks occurred,
namely in Angola, Haiti, Zimbabwe and Yemen \cite{AsianScientist,Shuai,WHO}.
Mathematical models have been developed and studied
in order to understand the transmission dynamics of cholera epidemics, 
mostly focusing on the epidemic that occurred in Haiti, 2010--2011 \cite{AsianScientist}. 
In \cite{Lemos-Paiao}, a SIQR (Susceptible--Infectious--Quarantined--Recovered) 
type model, which also considers a class for bacterial concentration, is analysed,  
being shown that it fits well the cholera outbreak in the Department 
of Artibonite -- Haiti, from 1$^{\text{st}}$ November 2010 until 1$^{\text{st}}$ May 2011. 
Furthermore, an optimal control problem is formulated and solved, 
where the optimal control function represents the fraction of infected 
individuals that are submitted to treatment through quarantine \cite{Lemos-Paiao}.
Many studies have been developed with the purpose to find and evaluate 
measures to contain the cholera spread. Nevertheless, it was not possible yet
to obtain solutions in real time that can stop the 
cholera epidemics \cite{AsianScientist}.

Recently, the biggest outbreak of cholera in the history of the world 
has occurred in Yemen \cite{TelegraphNews}. 
The epidemic began in October 2016 and in February-March 2017 was in decline. 
However, on 27$^{\text{th}}$ April 2017 the epidemic returned. This happened ten 
days after Sana'a's sewer system had stopped working. Problems in infrastructures, 
health, water and sanitation systems in Yemen, allowed the fast spread of the 
disease \cite{wikipedia_cholera_yemen}. Between 27$^{\text{th}}$ 
April 2017 and 15$^{\text{th}}$ April 2018 there were 1 090 280 
suspected cases reported and 2 275 deaths due to cholera \cite{WHO_15abril2018}.
In \cite{Nishiura}, Nishiura et al. study mathematically this outbreak, trying 
to forecast the cholera epidemic in Yemen, explicitly 
addressing the reporting delay and ascertainment bias.
A vaccine for cholera is currently available, but poor sanitation and the 
lack of access to vaccines promote the spread of the disease \cite{AsianScientist}.

Following World Health Organization (WHO) recommendations, the first-ever 
oral vaccination campaign against cholera had been launched on 6$^\text{th}$ 
May 2018 in Yemen, but it was concluded on 15$^\text{th}$ May 2018 \cite{WHO_vaccination}, 
due to the lack of national governmental authorization to do the vaccination \cite{reuters}. 
Aid workers say that one reason of the delayed of the campaign vaccination 
is due to some senior Houthi officials who objected to vaccination. 
This campaign coincided with the rainy season and some health workers 
fear that this could spread the disease \cite{reuters}. 
This campaign just covered four districts in Aden, which were at a
high risk of fast spread of the disease, and just 350 000 individuals 
(including pregnant women) were vaccinated \cite{reuters,TelegraphNews}. 
It is important to note that the vaccinated individuals represent, approximately, 
1.21\% of the total population, since the total population of Yemen, in 2018, 
is 28 915 284 \cite{Yemen_pop}. Lorenzo Pizzoli, WHO's cholera expert, said that the 
campaign hoped to cover at least four million people in areas at risk (corresponding, 
approximately, to 14\% of total population) and Michael Ryan, WHO's Assistant 
Director-General, revealed that they were negotiating with Yemen health authorities 
in order to vaccinate people from all high risks zones \cite{reuters}. 
The International Coordinating Group on Vaccine Provision had planned 
one million cholera vaccines for Yemen in July 2017, but WHO and Yemen 
local authorities decided to postpone it and the doses were diverted 
to South Sudan. WHO and GAVI, the Vaccine Alliance, affirmed that the largest 
cholera vaccination is now being carried out in five countries 
(Kenya, Malawi, South Sudan, Uganda and Zambia). It is expected that 
this campaign targets more than two million people across Africa \cite{reuters}.

In this paper, we propose a SITRV (Susceptible--Infectious--Treated--Recovered--Vaccinated) 
type model, which includes a class of bacterial concentration. In Section~\ref{Sec:model},
we formulate and explain the mathematical model. Then, in Section~\ref{Sec:mod:analysis_vaccination}, 
we show that the model is mathematically well posed and it has biological meaning. 
We prove the existence and uniqueness of the disease-free and endemic equilibrium points 
and compute the basic reproduction number. The sensitivity of the basic reproduction 
number with respect to all parameters of the model is analysed in 
Subsection~\ref{subsec:sensitivity_r0} and the stability analysis 
of equilibria is carried out in Subsection~\ref{subsec:stability_analysis}. 
In Section~\ref{sec:num:simu_vaccination}, we show that the model fits well 
the cholera outbreak in Yemen, between 27$^{\text{th}}$ April 2017 and 
15$^{\text{th}}$ April 2018. Through numerical simulations, we illustrate the 
impact of vaccination of susceptible individuals in Yemen. We end with 
Section~\ref{sec:conclusion}, by deriving some conclusions about the 
importance of vaccination campaigns on the control and eradication 
of a cholera outbreak.

We trust that our work is of great significance, because 
it provides a mathematical model for cholera that is deeply studied  
and allows to obtain important conclusions about the relevance 
of vaccination campaigns in cholera outbreaks. We show that, 
if it had existed a vaccination campaign from the beginning 
of the outbreak in Yemen, then the epidemic would have been extinguished. 
Actually, we believe that the absence of this type of prevention 
measures in Yemen, was one of the responsible for provoking the biggest 
cholera outbreak in world's history \cite{TelegraphNews}, killing 
2275 individuals until 15$^{\text{th}}$ April 2018.


\section{Model formulation}
\label{Sec:model}

We modify the model studied in \cite{Lemos-Paiao}, 
adding a vaccination class and considering different kinds of cholera's treatment.
The model is a SITRV (Susceptible--Infectious--Treated--Recovered--Vaccinated) 
type model and considers a class of bacterial concentration for the 
dynamics of cholera. The total human population $N(t)$ is divided into 
five classes: susceptible $S(t)$, infectious with symptoms $I(t)$, 
in treatment $T(t)$, recovered $R(t)$ and vaccinated $V(t)$ at time $t$, 
for $t \ge 0$. Furthermore, we consider a class $B(t)$ that reflects 
the bacterial concentration at time $t$. 
We assume that there is a positive 
recruitment rate $\Lambda$ into the susceptible class $S(t)$ and a positive 
natural death rate $\mu$, for all time $t$ under study. Susceptible
individuals can be vaccinated at rate $\varphi\geq 0$ and become infected 
with cholera at rate $\displaystyle \frac{\beta B(t)}{\kappa+B(t)}\geq 0$,
which is dependent on time $t$. Note that $\beta>0$ is the ingestion 
rate of the bacterium through contaminated sources, $\kappa>0$ is the half 
saturation constant of the bacterial population, and $\displaystyle \frac{B(t)}{\kappa+B(t)}$ 
measures the possibility of an infected individual to have the disease with symptoms, 
given a contact with contaminated sources \cite{Mwasa}. Any recovered 
and vaccinated individual can lose the immunity at rate $\omega_1\geq 0$ 
and $\omega_2\geq0$, respectively, becoming susceptible again. 
Different types of treatment for cholera infected individuals 
are considered based on \cite{WHO_treatment}.
The infected individuals can get a proper treatment, 
at rate $\delta\geq 0$, and the individuals in treatment can recover at rate
$\varepsilon\geq0$. The disease-related death rates associated with the individuals
that are infected and in treatment are $\alpha_1\geq0$ and $\alpha_2\geq0$, respectively. 
Each infected individual contributes to the increase of the bacterial concentration
at rate $\eta>0$. On the other hand, the bacterial concentration can decrease
at mortality rate $d>0$. These assumptions are translated into the following
mathematical model:
\begin{align}
\label{ModeloColera_vaccination}
&\begin{cases}
S'(t)=\Lambda-\displaystyle\frac{\beta B(t)}{\kappa+B(t)}S(t)
+\omega_1 R(t)+\omega_2V(t)-(\varphi+\mu)S(t),\\
I'(t)=\displaystyle\frac{\beta B(t)}{\kappa+B(t)}S(t)
-\left(\delta+\alpha_1+\mu\right)I(t),\\
T'(t)=\delta I(t)-(\varepsilon+\alpha_2+\mu)T(t),\\
R'(t)=\varepsilon T(t)-(\omega_1+\mu)R(t),\\
V'(t)=\varphi S(t)-(\omega_2+\mu)V(t),\\
B'(t)=\eta I(t)-dB(t).
\end{cases}
\end{align}


\section{Model analysis}
\label{Sec:mod:analysis_vaccination}

Throughout the paper, we assume that the initial conditions 
of system \eqref{ModeloColera_vaccination} are non-negative:
\begin{equation}
\label{eq:init:cond_vaccination}
S(0) = S_0 \geq 0, \  I(0) = I_0 \geq 0, 
\  T(0)= T_0 \geq 0, \  R(0) = R_0 \geq 0, 
\  V(0) = V_0 \geq 0, \  B(0) = B_0 \geq 0. 
\end{equation}


\subsection{Positivity and boundedness of solutions}
\label{subsec:positivity_boundedness}

Our first lemma shows that the considered model 
\eqref{ModeloColera_vaccination}--\eqref{eq:init:cond_vaccination} 
is biologically meaningful.

\begin{lemma}
The solutions $(S(t), I(t), T(t), R(t), V(t), B(t))$ 
of system \eqref{ModeloColera_vaccination} 
are non-negative for all $t \geq 0$ 
with non-negative initial conditions 
\eqref{eq:init:cond_vaccination} in $(\mathbb{R}_0^+)^6$. 
\end{lemma}

\begin{proof}
We have
\begin{equation*}
\begin{cases}
\frac{d S(t)}{d t}\Bigg|_{\xi(S)} &= \Lambda + \omega_1 R(t) + \omega_2 V(t)> 0 \, , \\
\frac{d I(t)}{d t}\Bigg|_{\xi(I)} &= \frac{\beta B(t)}{\kappa+B(t)}S(t) \geq 0 \, , \\
\frac{d T(t)}{d t}\Bigg|_{\xi(Q)} &= \delta I(t) \geq 0 \, , \\
\frac{d R(t)}{d t}\Bigg|_{\xi(R)} &= \varepsilon T(t) \geq 0 \, , \\
\frac{d V(t)}{d t}\Bigg|_{\xi(V)} &= \varphi S(t) \geq 0 \, , \\
\frac{d B(t)}{d t}\Bigg|_{\xi(B)} &= \eta I(t) \geq 0 \, ,
\end{cases}
\end{equation*}
where $\xi(\upsilon)=\left\{\upsilon(t)=0 \text{ and } 
S, I, T, R, V, B \in C(\mathbb{R}_0^+,\mathbb{R}_0^+)\right\}$
and $\upsilon \in \{S, I, T, R, V, B\}$.
Therefore, due to Lemma~2 in \cite{Yang:CMA:1996}, 
any solution of system \eqref{ModeloColera_vaccination} is such that 
$(S(t), I(t), T(t), R(t), V(t), B(t)) \in (\mathbb{R}_0^+)^6$
for all $t \geq 0$. 
\end{proof}

Lemma~\ref{lema:Inv_vaccination} shows that it is enough 
to consider the dynamics of the flow generated by 
\eqref{ModeloColera_vaccination}--\eqref{eq:init:cond_vaccination} 
in a certain region $\Omega_V$. 

\begin{lemma}
\label{lema:Inv_vaccination}
Let 
\begin{equation}
\label{eqN:1_vaccination}
\Omega_H = \left\{ (S, I, T, R, V) \in \left(\mathbb{R}_0^+\right)^5 \, | \, 0 
\leq S(t) + I(t) + T(t) + R(t) + V(t)\leq \frac{\Lambda}{\mu} \right\} 
\end{equation}
and
\begin{equation}
\label{eq:maj:B}
\Omega_B = \left\{ B \in \mathbb{R}_0^+ \, | \, 0 \leq B(t)  
\leq \frac{\Lambda\eta}{\mu d} \right\}.
\end{equation}
Define 
\begin{equation}
\label{eq:Omega_vaccination}
\Omega_V = \Omega_H \times \Omega_B.
\end{equation}
If $N(0) \leq \displaystyle \frac{\Lambda}{\mu}$ and 
$B(0) \leq \displaystyle \frac{\Lambda\eta}{\mu d}$, 
then the region $\Omega_V$ is positively invariant for model \eqref{ModeloColera_vaccination} 
with non-negative initial conditions \eqref{eq:init:cond_vaccination} in $(\mathbb{R}_0^+)^6$.
\end{lemma}

\begin{proof}
Let us split system \eqref{ModeloColera_vaccination} into two parts: 
the human population, i.e., $S(t)$, $I(t)$, $T(t)$, $R(t)$ and $V(t)$, 
and the pathogen population, i.e., $B(t)$.
Adding the first five equations of system \eqref{ModeloColera_vaccination} gives
\begin{equation*}
N'(t) = S'(t)+I'(t)+T'(t)+R'(t)+V'(t)=\Lambda-\mu N(t)
-\alpha_1 I(t)-\alpha_2 T(t) \leq \Lambda-\mu N(t) \, .
\end{equation*}
Assuming that $N(0) \leq \frac{\Lambda}{\mu}$, we conclude
that $N(t) \leq  \frac{\Lambda}{\mu}$. For this reason,
\eqref{eqN:1_vaccination} defines the biologically feasible region 
for the human population. As it is proved in \cite{Lemos-Paiao}, 
the region \eqref{eq:maj:B} defines the biologically feasible region 
for the pathogen population. From \eqref{eqN:1_vaccination} and \eqref{eq:maj:B}, 
we know that $N(t)$ and $B(t)$ are bounded for all $t \geq 0$. 
Therefore, every solution of system \eqref{ModeloColera_vaccination} 
with initial conditions in $\Omega_V$ remains in $\Omega_V$ for all $t \geq 0$. 
In other words, in region $\Omega_V$ defined by \eqref{eq:Omega_vaccination}, 
our model is epidemiologically and mathematically well posed 
in the sense of \cite{Hethcote}. 
\end{proof}


\subsection{Equilibrium points and basic reproduction number}
\label{subsec:eqpoints_r0}

From now on, let us consider that $a_0=\varphi+\mu$, $a_1=\delta+\alpha_1+\mu$, 
$a_2=\varepsilon+\alpha_2+\mu$, $a_3=\omega_1+\mu$ and $a_4=\omega_2+\mu$. 
The disease-free equilibrium (DFE) of model \eqref{ModeloColera_vaccination} 
is given by
\begin{equation}
\label{eq:DFE_vaccination}
E^0=\left(S^0,I^0,T^0,R^0,V^0,B^0\right)
=\left(\frac{\Lambda a_4}{a_0a_4-\varphi\omega_2},0,0,0,
\frac{\Lambda\varphi}{a_0a_4-\varphi\omega_2},0\right).
\end{equation}

\begin{remark}
\label{remark:a0a4_positivos}
Note that, because $\mu>0$, one has 
$a_0a_4-\varphi\omega_2=(\varphi+\mu)(\omega_2+\mu)-\varphi\omega_2>0$. 
\end{remark}

Next, following the approach of \cite{Mwasa,Driessche},
we compute the basic reproduction number $R_0$.

\begin{proposition}[Basic reproduction number of \eqref{ModeloColera_vaccination}]
\label{prop:R0_vaccination}
The basic reproduction number of model \eqref{ModeloColera_vaccination} is given by
\begin{align}
\label{R0_vaccination}
R_0=\frac{\beta\Lambda\eta a_4}{(a_0a_4-\varphi\omega_2)\kappa da_1}.
\end{align}
\end{proposition}

\begin{proof}
Consider that $\mathcal{F}_i(t)$ is the rate of appearance 
of new infections in the compartment associated with index $i$, 
$\mathcal{V}_i^+(t)$ is the rate of transfer of ``individuals'' 
into the compartment associated with index $i$ by all other means 
and $\mathcal{V}_i^-(t)$ is the rate of transfer of ``individuals'' 
out of compartment associated with index $i$. In this way, the matrices 
$\mathcal{F}(t)$, $\mathcal{V}^+(t)$ and $\mathcal{V}^-(t)$, 
associated with model \eqref{ModeloColera_vaccination}, are given by
\begin{align*}
\mathcal{F}(t)=\left[
\begin{matrix}
0\\
\displaystyle\frac{\beta B(t)S(t)}{\kappa+B(t)}\\
0\\
0\\
0\\
0
\end{matrix}
\right],
\quad \mathcal{V}^+(t)=\left[
\begin{matrix}
\Lambda+\omega_1 R(t)+\omega_2V(t)\\
0\\
\delta I(t)\\
\varepsilon T(t)\\
\varphi S(t)\\
\eta I(t)
\end{matrix}
\right]
\quad \text{ and } \quad
\mathcal{V}^-(t)=\left[
\begin{matrix}
\displaystyle\frac{\beta B(t)S(t)}{\kappa+B(t)}+a_0S(t)\\
a_1 I(t)\\
a_2 T(t)\\
a_3 R(t)\\
a_4 V(t)\\
d B(t)
\end{matrix}\right].
\end{align*}
Therefore, by considering 
$\mathcal{V}(t)=\mathcal{V}^-(t)-\mathcal{V}^+(t)$,
we have that 
\begin{align*}
\left[
\begin{matrix}
S'(t) & I'(t) & T'(t) & R'(t) & V'(t) & B'(t)
\end{matrix}\right]^T=\mathcal{F}(t)-\mathcal{V}(t).
\end{align*}
The Jacobian matrices of $\mathcal{F}(t)$ 
and of $\mathcal{V}(t)$ are, respectively, given by
\begin{align*}
F=\left[\begin{matrix}
0 & 0 & 0 & 0 & 0 & 0 \\
\displaystyle\frac{\beta B(t)}{\kappa+B(t)} 
& 0 & 0 & 0 & 0 &
\displaystyle\frac{\beta\kappa S(t)}{(\kappa+B(t))^2} \\
0 & 0 & 0 & 0 & 0 & 0 \\
0 & 0 & 0 & 0 & 0 & 0 \\
0 & 0 & 0 & 0 & 0 & 0 \\
0 & 0 & 0 & 0 & 0 & 0
\end{matrix}\right]
\text{ and }
V=\left[\begin{matrix}
\displaystyle\frac{\beta B(t)}{\kappa+B(t)}+a_0 
& 0 & 0 & -\omega_1 & -\omega_2 & 
\displaystyle\frac{\beta\kappa S(t)}{(\kappa+B(t))^2} \\
0 & a_1 & 0 & 0 & 0 & 0 \\
0 & -\delta & a_2 & 0 & 0 & 0 \\
0 & 0 & -\varepsilon & a_3 & 0 & 0 \\
-\varphi & 0 & 0 & 0 & a_4 & 0 \\
0 & -\eta & 0 & 0 & 0 & d \\
\end{matrix}\right].
\end{align*}
In the disease-free equilibrium $E^0$ defined by \eqref{eq:DFE_vaccination},  
we obtain the matrices $F_0$ and $V_0$ given by
\begin{align*}
F_0=
\left[\begin{matrix}
0 & 0 & 0 & 0 & 0 & 0 \\
0 & 0 & 0 & 0 & 0 &
\displaystyle\frac{\beta\Lambda a_4}{(a_0a_4-\varphi\omega_2)\kappa} \\
0 & 0 & 0 & 0 & 0 & 0 \\
0 & 0 & 0 & 0 & 0 & 0 \\
0 & 0 & 0 & 0 & 0 & 0 \\
0 & 0 & 0 & 0 & 0 & 0
\end{matrix}\right]
\quad \text{ and } \quad
V_0=
\left[\begin{matrix}
a_0 & 0 & 0 & -\omega_1 & -\omega_2 & 
\displaystyle\frac{\beta\Lambda a_4}{(a_0a_4-\varphi\omega_2)\kappa} \\
0 & a_1 & 0 & 0 & 0 & 0 \\
0 & -\delta & a_2 & 0 & 0 & 0 \\
0 & 0 & -\varepsilon & a_3 & 0 & 0 \\
-\varphi & 0 & 0 & 0 & a_4 & 0 \\
0 & -\eta & 0 & 0 & 0 & d \\
\end{matrix}\right].
\end{align*}
The basic reproduction number 
of model \eqref{ModeloColera_vaccination} is then given by
\begin{align*}
R_0=\rho(F_0V_0^{-1})
=\frac{\beta\Lambda\eta a_4}{(a_0a_4-\varphi\omega_2)\kappa da_1},
\end{align*}
found with the help of the computer algebra system \textsf{Maple}.
This concludes the proof.
\end{proof}

Now we prove the existence of an endemic equilibrium when 
$R_0$ given by \eqref{R0_vaccination} is greater than one.

\begin{proposition}[Endemic equilibrium]
\label{prop:EE_vaccination}
If the basic reproduction number \eqref{R0_vaccination}
is such that $R_0>1$, then the model \eqref{ModeloColera_vaccination} 
has an endemic equilibrium given by
\begin{equation}
\label{EndemicEquilibrium_vaccination}
E^*=(S^*,I^*,T^*,R^*,V^*,B^*),
\end{equation}
where
\begin{equation*}
\begin{cases}
S^*=\displaystyle\frac{a_1a_4\{\kappa d(a_1a_2a_3-\delta\varepsilon\omega_1)
+\Lambda\eta a_2a_3\}}{\eta\tilde{D}}\\[0.3cm]
I^*=\displaystyle\frac{a_2a_3\{\beta\Lambda\eta a_4
-(a_0a_4-\varphi\omega_2)\kappa da_1\}}{\eta\tilde{D}}\\[0.3cm]
T^*=\displaystyle\frac{a_3\delta\{\beta\Lambda\eta a_4
-(a_0a_4-\varphi\omega_2)\kappa da_1\}}{\eta\tilde{D}}\\[0.3cm]
R^*=\displaystyle\frac{\delta\epsilon\{\beta\Lambda\eta a_4
-(a_0a_4-\varphi\omega_2)\kappa da_1\}}{\eta\tilde{D}}\\[0.3cm]
V^*=\displaystyle\frac{a_1\varphi\{\kappa d(a_1a_2a_3
-\delta\varepsilon\omega_1)+\Lambda\eta a_2a_3\}}{\eta\tilde{D}}\\[0.3cm]
B^*=\displaystyle\frac{a_2a_3\{\beta\Lambda\eta a_4
-(a_0a_4-\varphi\omega_2)\kappa da_1\}}{d\tilde{D}}
\end{cases}
\end{equation*}
and $\tilde{D}=\left\{a_1a_2a_3(a_0a_4-\varphi\omega_2)
+\beta a_4(a_1a_2a_3-\delta\varepsilon\omega_1)\right\}$.
\end{proposition}

\begin{proof}
We note that
\begin{enumerate}
\item $a_1=\delta+\alpha_1+\mu>0$, because $\delta$, $\alpha_1\geq0$ and $\mu>0$;
\item $a_2=\varepsilon+\alpha_2+\mu>0$, because $\varepsilon$, $\alpha_2\geq0$ and $\mu>0$;
\item $a_3=\omega_1+\mu>0$, because $\omega_1\geq0$ and $\mu>0$;
\item $\beta$, $\kappa$, $d>0$ and $\varphi\geq0$;
\item $a_0a_4-\varphi\omega_2>0$ (see Remark~\ref{remark:a0a4_positivos});
\item $a_1a_2a_3-\delta\varepsilon\omega_1=(\delta+\alpha_1
+\mu)(\varepsilon+\alpha_2+\mu)(\omega_1+\mu)-\delta\varepsilon\omega_1>0$, 
because $\alpha_1$, $\alpha_2\geq0$ and $\mu>0$;
\item $\Lambda\eta a_2a_3>0$, because $\Lambda$, $\eta$, $a_2$, $a_3>0$.
\end{enumerate}	
With the above inequalities, we conclude that $\tilde{D}>0$ 
and, consequently, that $S^*>0$ and $V^*\geq0$.
The basic reproduction number is given by 
$\displaystyle\frac{\beta\Lambda\eta a_4}{(a_0
a_4-\varphi\omega_2)\kappa da_1}$. Thus, it follows that
\begin{align*}
&\ \beta\Lambda\eta a_4=R_0(a_0a_4-\varphi\omega_2)\kappa da_1\\
\Leftrightarrow &\ \beta\Lambda\eta a_4-(a_0a_4-\varphi\omega_2)\kappa da_1
=R_0(a_0a_4-\varphi\omega_2)\kappa da_1-(a_0a_4-\varphi\omega_2)\kappa da_1\\
\Leftrightarrow &\ \beta\Lambda\eta a_4-(a_0a_4-\varphi\omega_2)\kappa da_1
=(a_0a_4-\varphi\omega_2)\kappa da_1(R_0-1).
\end{align*}
Therefore, we have that
\begin{equation*}
\begin{cases}
I^*=\displaystyle\frac{a_1a_2a_3\kappa 
d(a_0a_4-\varphi\omega_2)(R_0-1)}{\eta\tilde{D}}\\[0.3cm]
T^*=\displaystyle\frac{a_1a_3\kappa d\delta(a_0 a_4
-\varphi\omega_2)(R_0-1)}{\eta\tilde{D}}\\[0.3cm]
R^*=\displaystyle\frac{a_1\kappa d\delta\varepsilon(a_0
a_4-\varphi\omega_2)(R_0-1)}{\eta\tilde{D}}\\[0.3cm]
B^*=\displaystyle\frac{a_1a_2a_3\kappa(a_0
a_4-\varphi\omega_2)(R_0-1)}{\tilde{D}}.
\end{cases}
\end{equation*}
In order to obtain an endemic equilibrium, we have to ensure that $I^*$, $B^*>0$. 
Thus, we obtain $I^*$, $B^*>0$ if and only if $R_0-1>0\Leftrightarrow R_0>1$. 
In this case ($R_0>1$), we also have that $T^*$, $R^*\geq0$. 
\end{proof}


\subsection{Sensitivity of the basic reproduction number}
\label{subsec:sensitivity_r0}

In this section, we are going to study the sensitivity of $R_0$ 
with respect to all parameters $p$ of model \eqref{ModeloColera_vaccination}, 
computing the respective normalized forward sensitive indexes $\Upsilon^{R_0}_p$ given in
Definition~\ref{def:sensitivity_indexes}. They are presented in Table~\ref{Tab:R0_sensitivity}.

\begin{definition}[See \cite{Chitnis,Kong,Silva}]
\label{def:sensitivity_indexes}
The normalized forward sensitivity index of a variable $z$ that depends 
differentiably on a parameter $p$ is defined by
\begin{equation*}
\Upsilon^z_p=\frac{\partial z}{\partial p}\times\frac{p}{|z|}.
\end{equation*}
\end{definition}

\begin{remark}
When a parameter $p$ is one of the most sensitive parameters with respect to a variable $z$, 
then we have $\Upsilon^z_p=\pm1$. If $\Upsilon^z_p=1$, then an increase (decrease) 
of $p$ by $\gamma\%$ provokes an increase (decrease) of $z$ by $\gamma\%$. On the other hand, 
if $\Upsilon^z_p=-1$, then an increase (decrease) of $p$ by $\gamma\%$ provokes 
a decrease (increase) of $z$ by $\gamma\%$ (see \cite{Silva}).
\end{remark}
\begin{table}[ht!]
\doublerulesep 0.1pt
\tabcolsep 7.8mm
\centering
\caption{\rm The normalized forward sensitivity indexes $\Upsilon^{R_0}_p$ 
with respect to all parameters of model 
\eqref{ModeloColera_vaccination}.}\label{Tab:R0_sensitivity}
\vspace*{2mm}
\renewcommand{\arraystretch}{1.3}
\setlength{\tabcolsep}{10pt}
\footnotesize
{\begin{tabular*}{6.5cm}{cc}
\hline\hline\hline
\raisebox{-2ex}[0pt][0pt]{Parameter $p$} & \raisebox{-2ex}[0pt][0pt]{$\Upsilon^{R_0}_p$}\\
\\
\hline
\small $\Lambda$ & \small 1 \\
\small $\mu$ & \small $\mu\left(\frac{1}{a_4}-\frac{1}{a_1}
-\frac{\varphi+\omega_2+2\mu}{a_0a_4-\varphi\omega_2}\right)$ \\
\small $\beta$ & \small 1 \\
\small $\kappa$ & \small -1 \\
\small $\omega_1$ & \small 0 \\
\small $\omega_2$ & \small $\frac{\varphi\omega_2}{a_4(\varphi+\omega_2+\mu)}$ \\
\small $\varphi$ & \small $-\frac{\varphi}{\varphi+\omega_2+\mu}$ \\
\small $\delta$ & \small $-\frac{\delta}{a_1}$ \\
\small $\varepsilon$ & \small 0 \\
\small $\alpha_1$ & \small $-\frac{\alpha_1}{a_1}$ \\
\small $\alpha_2$ & \small  0 \\
\small $\eta$ & \small 1 \\
\small $d$ & \small -1 \\
\hline\hline\hline
\end{tabular*}
}
\renewcommand{\arraystretch}{1}
\end{table}


\subsection{Stability analysis}
\label{subsec:stability_analysis}

Now we prove the local stability of the disease-free equilibrium $E^0$. 

\begin{theorem}[Stability of the DFE \eqref{eq:DFE_vaccination}] 
The disease-free equilibrium $E^0$ of model \eqref{ModeloColera_vaccination} is
\begin{enumerate}
\item locally asymptotic stable, if $R_0<1$;
\item unstable, if $R_0>1$.
\end{enumerate}
Moreover, if $R_0=1$, then a critical case occurs.
\end{theorem}

\begin{proof}
The characteristic polynomial associated with the linearised system 
of model \eqref{ModeloColera_vaccination} is given by
\begin{align*}
p_V(\chi)=\det(F_0-V_0-\chi I_6).
\end{align*} 
In order to compute the roots of the polynomial $p_V$, we have that
\begin{equation*}
\left|\begin{matrix}
-a_0-\chi & 0 & 0 & \omega_1 & \omega_2 & 
-\displaystyle\frac{\beta \Lambda a_4}{(a_0a_4-\varphi\omega_2)\kappa}\\[0.25cm]
0 & -a_1-\chi & 0 & 0 & 0 & \displaystyle
\frac{\beta \Lambda a_4}{(a_0a_4-\varphi\omega_2)\kappa} \\[0.25cm]
0 & \delta & -a_2-\chi & 0 & 0 & 0\\[0.25cm]
0 & 0 & \varepsilon & -a_3-\chi & 0 & 0\\[0.25cm]
\varphi & 0 & 0 & 0 & -a_4-\chi & 0\\[0.25cm]
0 & \eta & 0 & 0 & 0 & -d-\chi
\end{matrix}\right|=0,
\end{equation*}
that is,
$$
\chi^2+(a_0+a_4)\chi+(a_0a_4-\varphi\omega_2)=0 \vee \chi^2
+(a_1+d)\chi+a_1d-\frac{\beta\Lambda\eta a_4}{(a_0a_4-\varphi\omega_2)\kappa}
=0 \vee \chi=-a_2 \vee \chi=-a_3.
$$
As the coefficients of polynomial $\chi^2+(a_0+a_4)\chi+(a_0a_4-\varphi\omega_2)$ 
have the same sign (see Remark~\ref{remark:a0a4_positivos}), then it follows from
Routh's criterion that their roots have negative real part (see, e.g., pp.~55--56 of \cite{Olsder}). 
Furthermore, using similar arguments, the roots of the polynomial
$$
\chi^2+(a_1+d)\chi+a_1d
-\frac{\beta\Lambda\eta a_4}{(a_0a_4-\varphi\omega_2)\kappa}
$$
have negative real part if and only if
$$
a_1d-\frac{\beta\Lambda\eta a_4}{(a_0a_4-\varphi\omega_2)\kappa}
>0 \Leftrightarrow R_0<1.
$$
Therefore, the DFE $E^0$ is (i) locally asymptotic stable if $R_0<1$;
(ii) unstable if $R_0>1$; (iii) critical if $R_0=1$. 
\end{proof}

We end this section by proving the local stability of the endemic 
equilibrium $E^*$. Our proof is based on the center manifold 
theory \cite{Carr}, as described in Theorem~4.1 of \cite{Castillo}.

\begin{theorem}[Local asymptotic stability of the endemic 
equilibrium \eqref{EndemicEquilibrium_vaccination}] 
The endemic equilibrium $E^*$ of model \eqref{ModeloColera_vaccination}
(see Proposition~\ref{prop:EE_vaccination}) is locally asymptotic stable for $R_0>1$.
\end{theorem}

\begin{proof}
In order to apply the method described in Theorem~4.1 of \cite{Castillo}, 
we are going to do the following change of variables. Let us consider that 
\begin{equation*}
X = (x_1, x_2, x_3, x_4, x_5, x_6) = (S, I, T, R, V, B).
\end{equation*} 
Consequently, we have that the total number of individuals 
is given by $N=\sum_{i=1}^{5}x_i$. Thus, the model 
\eqref{ModeloColera_vaccination} can be written as follows:
\begin{align}
\label{ModeloColeraX_vaccination}
\begin{cases}
x_1'(t)=f_1=\Lambda-\displaystyle\frac{\beta x_6(t)}{\kappa+x_6(t)}x_1(t)
+\omega_1 x_4(t)+\omega_2x_5(t)-a_0x_1(t),\\
x_2'(t)=f_2=\displaystyle\frac{\beta x_6(t)}{\kappa+x_6(t)}x_1(t)-a_1x_2(t),\\
x_3'(t)=f_3=\delta x_2(t)-a_2x_3(t),\\
x_4'(t)=f_4=\varepsilon x_3(t)-a_3x_4(t),\\
x_5'(t)=f_5=\varphi x_1(t)-a_4x_5(t),\\
x_6'(t)=f_6=\eta x_2(t)-dx_6(t).
\end{cases}
\end{align}
Choosing $\beta^*$ as bifurcation parameter and solving for $\beta$, 
from $R_0=1$ we have that
\begin{align*}
\beta^*=\frac{(a_0a_4-\varphi\omega_2)\kappa da_1}{\Lambda\eta a_4}.
\end{align*}
Considering $\beta=\beta^*$, the Jacobian of the 
system \eqref{ModeloColeraX_vaccination} 
evaluated at $E^0$ is given by 
\begin{align*}
J_0^*=\left[
\begin{matrix}
-a_0 & 0 & 0 & \omega_1 & \omega_2 
& -\displaystyle\frac{a_1d}{\eta} \\
0 & -a_1 & 0 & 0 & 0 & \displaystyle\frac{a_1d}{\eta} \\
0 & \delta & -a_2 & 0 & 0 & 0\\
0 & 0 & \varepsilon & -a_3 & 0 & 0\\
\varphi & 0 & 0 & 0 & -a_4 & 0\\
0 & \eta & 0 & 0 & 0 & -d \\
\end{matrix}
\right].
\end{align*}
The eigenvalues of $J_0^*$ are obtained solving the equation 
$\det(J_0^*-\chi I_6)=0$. Thus, we have that
\begin{align*}
&\ \det(J_0^*-\chi I_6)=0\\
\Leftrightarrow&\ \chi=0\ \vee\ \chi=-a_1-d\ \vee\ \chi
=-a_2\ \vee\ \chi=-a_3\ \vee\ \chi=-\frac{1}{2}\left(a_0+a_4
\pm\sqrt{(a_0-a_4)^2+4\varphi\omega_2}\right).
\end{align*}
Note that the eigenvalue $\chi=-\frac{1}{2}\left(a_0+a_4-\sqrt{(a_0-a_4)^2
+4\varphi\omega_2}\right)$ is a negative real number, because
\begin{equation*}
\begin{split}
-\frac{1}{2}\left(a_0+a_4-\sqrt{(a_0-a_4)^2+4\varphi\omega_2}\right)
&=-\frac{1}{2}\left(\varphi+\mu+\omega_2+\mu-\sqrt{(\varphi+\mu-\omega_2-\mu)^2
+4\varphi\omega_2}\right)\\
&=-\frac{1}{2}\left(\varphi+\omega_2+2\mu-\sqrt{\varphi^2-2\varphi\omega_2
+\omega_2^2+4\varphi\omega_2}\right)\\
&=-\frac{1}{2}\left(\varphi+\omega_2+2\mu-\sqrt{(\varphi+\omega_2)^2}\right)\\
&\underset{\varphi+\omega_2\geq0}{=}
-\frac{1}{2}\left(\varphi+\omega_2+2\mu-(\varphi+\omega_2)\right)\\
&=-\mu<0.
\end{split}
\end{equation*}
Therefore, we can conclude that a simple eigenvalue of $J_0^*$ is zero, while all other 
eigenvalues of $J_0^*$ have negative real parts. So, the center manifold
theory \cite{Carr} can be applied to study the dynamics of \eqref{ModeloColeraX_vaccination} 
near $\beta=\beta^*$. Theorem 4.1 in \cite{Castillo} is used to show 
the local asymptotic stability of the endemic equilibrium point of \eqref{ModeloColeraX_vaccination}, 
for $\beta$ near $\beta^*$. The Jacobian $J_0^*$ has, respectively, a right 
eigenvector and a left eigenvector (associated with the zero eigenvalue), 
$$
w=\left[
\begin{matrix}
w_1\ \ \ & w_2\ \ \ & w_3\ \ \ & w_4\ \ \ & w_5\ \ \ & w_6
\end{matrix}
\right]^T 
\quad \text{ and } \quad 
v=\left[
\begin{matrix}
v_1\ \ \ & v_2\ \ \ & v_3\ \ \ & v_4\ \ \ & v_5\ \ \ & v_6
\end{matrix}
\right]^T,
$$ 
given by
\begin{align*}
w&=\left[
\begin{matrix}
\displaystyle-\frac{a_4(a_1a_2a_3-\delta\varepsilon\omega_1)}{a_2
a_3(a_0a_4-\varphi\omega_2)}\ \ \ & 1\ \ \ & \displaystyle\frac{\delta}{a_2}
\ \ \ & \displaystyle\frac{\delta\varepsilon}{a_2a_3}\ \ \ 
& \displaystyle\frac{\varphi}{a_4}\ \ \ & \displaystyle\frac{\eta}{d}
\end{matrix}
\right]^T
w_2
\end{align*}
and
\begin{align*}
v=\left[\begin{matrix}0\ \ \ & 1\ \ \ & 0\ \ \ & 0\ \ \ & 0\ \ \ &
\displaystyle\frac{a_1}{\eta}\end{matrix}\right]^Tv_2.
\end{align*}
Remember that $f_l$ represents the right-hand side of the 
$l$th equation of the system \eqref{ModeloColeraX_vaccination} 
and $x_l$ is the state variable whose derivative is given 
by the $l$th equation, $l=1, \ldots, 6$. 
The local stability near the bifurcation point $\beta=\beta^*$ 
is determined by the signs of two associated constants 
$a$ and $b$ defined by
\begin{align*}
a=\sum_{k,i,j=1}^{6}v_kw_iw_j\left[\frac{\partial^2f_k}{\partial x_i
\partial x_j}(E^0)\right]_{\beta=\beta^*}
\text{\ \ and\ \ }
b=\sum_{k,i=1}^{6}v_kw_i\left[\frac{\partial^2f_k}{\partial x_i
\partial \phi}(E^0)\right]_{\beta=\beta^*}
\end{align*}
with $\phi=\beta-\beta^*$. As $v_1=v_3=v_4=v_5=0$, 
we only have to consider the following non-zero partial 
derivatives at the disease free equilibrium $E^0$:
\begin{align*}
\left[\frac{\partial^2f_2}{\partial x_1\partial x_6}(E^0)\right]_{\beta=\beta^*}
=\left[\frac{\partial^2f_2}{\partial x_6\partial x_1}(E^0)\right]_{\beta=\beta^*}
=\frac{\beta^*}{\kappa}
\ \text{\ \ and\ \ }\ \left[\frac{\partial^2f_2}{\partial x_6^2}(E^0)\right]_{\beta=\beta^*}
=-\frac{2\beta^*\Lambda a_4}{a_0a_4-\varphi\omega_2}.
\end{align*}
Therefore, the constant $a$ is
\begin{align*}
a=-\frac{2\beta^*\eta a_4}{d(a_0a_4-\varphi\omega_2)}\left(
\frac{a_1a_2a_3-\delta\varepsilon\omega_1}{a_2a_3\kappa}
+\frac{\Lambda\eta}{d}\right)v_2w_2^2<0.
\end{align*}
Furthermore, we have that
\begin{equation*}
\begin{split}
b=v_2w_6\left[\frac{\partial^2f}{\partial 
x_6\partial\phi}(E^0)\right]_{\beta=\beta^*}
=\frac{\Lambda\eta a_4}{\kappa d (a_0a_4-\varphi\omega_2)}v_2w_2>0.
\end{split}
\end{equation*}
Thus, as 
\begin{equation*}
\begin{cases}
a<0\\
b>0\\
\phi=\beta-\beta^*=\displaystyle\frac{a_1\kappa d(a_0a_4
-\varphi\omega_2)}{\Lambda\eta a_4}(R_0-1)>0
\end{cases}
\Leftrightarrow
\begin{cases}
a<0\\
b>0\\
R_0>1,
\end{cases}
\end{equation*}
we conclude from Theorem~4.1 in \cite{Castillo}
that the endemic equilibrium $E^*$ of \eqref{ModeloColera_vaccination} 
is locally asymptotic stable for a value of the basic reproduction number such that $R_0>1$. 
\end{proof}


\section{Numerical Simulations}
\label{sec:num:simu_vaccination}

In this section, we simulate the worst cholera outbreak
that ever occurred in human history. 
It occurred in Yemen, from 27$^{\text{th}}$ April 2017 
to 15$^{\text{th}}$ April 2018 \cite{WHO}. 
As the first-ever oral cholera vaccination campaign had been launched 
only on 6$^{\text{th}}$ May 2018 and was concluded on 15$^{\text{th}}$ 
May 2018 \cite{WHO_vaccination}, to describe such reality of Yemen, 
a numerical simulation of our model is carried out 
with $\varphi=\omega_2=V(0)=0$, that is, in absence of vaccination, 
and with all the other values as in Table~\ref{Tab_Parameter_vaccination}. 
We also simulate an hypothetical situation that includes vaccination 
from the beginning of the outbreak, considering in that case
all values of Table~\ref{Tab_Parameter_vaccination}. Let us denote the numerical 
simulation without and with vaccination by $(NS)$ and $(NS_V)$, respectively. 
The curves of infected individuals for $(NS)$ and $(NS_V)$ 
can be observed in Figure~\ref{fig:cholera_yemen_with_without_vaccination}, 
respectively in blue solid line and in blue dashed line. 
Our results allow us to state that if a vaccination campaign 
had been considered earlier in time, the number of infected individuals 
would have been significantly lower. Furthermore, the basic reproduction number 
of the simulation without vaccination is $R_0\simeq6.132305>1$ and the one with 
vaccination is $R_0\simeq0.753969<1$. This means that if vaccination 
had been considered from the beginning of the outbreak, then the spread of 
cholera would have been extinguished. Consequently, there would not 
have been so many deaths. Note that the decrease of $R_0$ with the introduction 
of a vaccination campaign is expected, because 
$$
\Upsilon^{R_0}_\varphi=-\displaystyle
\frac{\varphi}{\varphi+\omega_2+\mu}\simeq-0.877050<0.
$$
Furthermore, for $(NS)$, we obtain an endemic equilibrium point given by
\begin{multline*}
(S^*, I^*, T^*, R^*, V^*, B^*)\\
=(2.943350\times10^7,\  1.035599
\times10^5,\  5.954131\times10^5,\  1.070992\times10^8,\  0,\  3.138180\times10^6)
\end{multline*}
and for $(NS_V)$ we have a disease-free equilibrium point given by
$$
(S^0, I^0, T^0, R^0, V^0, B^0)
=(1.689119\times10^7,\  0,\  0,\  0,\  1.204910\times10^8,\  0).
$$
Note that the previous figures correspond to the equilibrium points 
for the parameter values of Table~\ref{Tab_Parameter_vaccination}, 
which can be obtained numerically for a final time of approximately 1370 years. 
We also call attention to the fact that the recruitment rate $\Lambda$
of Yemen is big and this leads to a huge growth of the population.
\begin{figure}[ht!]
\centering
\includegraphics[scale=0.47]{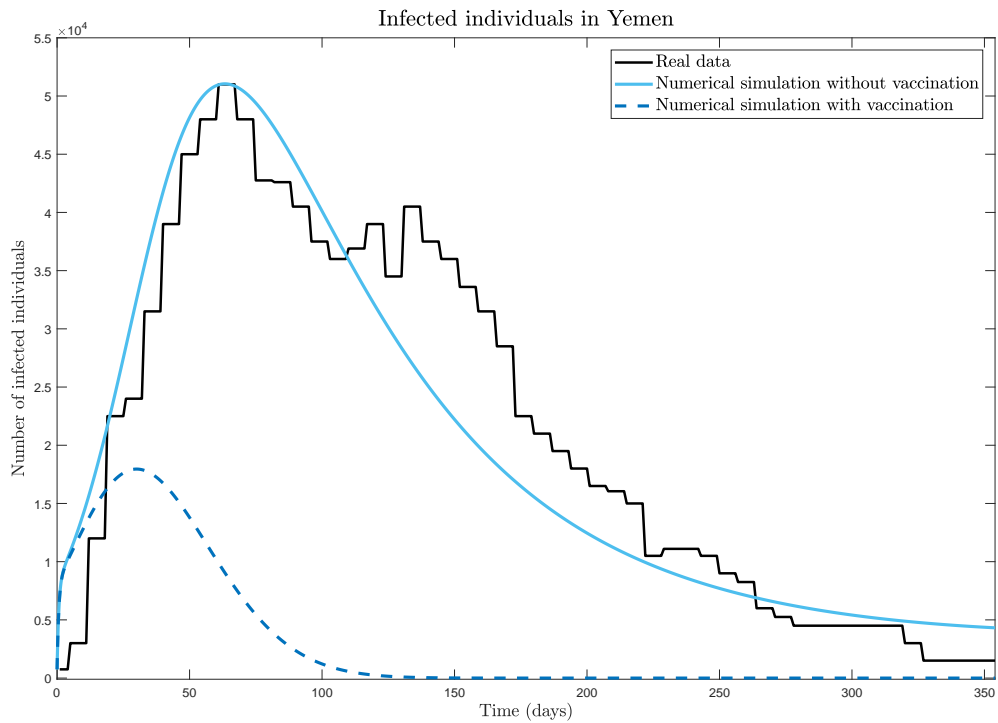}
\caption{Numerical solution of model 
\eqref{ModeloColera_vaccination}--\eqref{eq:init:cond_vaccination}
with the parameters of Table~\ref{Tab_Parameter_vaccination}
without (blue solid line) and with 
(blue dashed line) vaccination and real data (black solid line) 
of infected individuals $I(t)$ in Yemen from 27$^{\text{th}}$ 
April 2017 to 15$^{\text{th}}$ April 2018.}
\label{fig:cholera_yemen_with_without_vaccination}
\end{figure}
\begin{table}[ht!]
\doublerulesep 0.1pt
\tabcolsep 7.8mm
\centering
\caption{\rm Parameter values and initial conditions for the SIQRVB model 
\eqref{ModeloColera_vaccination}.}\label{Tab_Parameter_vaccination}
\vspace*{2mm}
\renewcommand{\arraystretch}{1.3}
\setlength{\tabcolsep}{10pt}
\footnotesize
{\begin{tabular*}{15.5cm}{cccc}
\hline\hline\hline
\raisebox{-2ex}[0pt][0pt]{Parameter} & \raisebox{-2ex}[0pt][0pt]{Description} 
& \raisebox{-2ex}[0pt][0pt]{Value} & \raisebox{-2ex}[0pt][0pt]{Reference} \\
\\
\hline
\small $\Lambda$ & \small Recruitment rate 
& \small 28.4$N(0)$/365000 (day$^{-1}$) & \small \cite{BirthRate}\\
\small $\mu$ & \small Natural death rate & \small 
1.6$\times10^{-5}$ (day$^{-1}$) & \small \cite{DeathRate}\\
\small $\beta$ & \small Ingestion rate 
& \small 0.01694 (day$^{-1}$)  & \small Assumed\\
\small $\kappa$ & \small Half saturation constant 
& \small$10^7$ (cell/ml) & \small Assumed\\
\small $\omega_1$ & \small Immunity waning rate 
& \small 0.4/365 (day$^{-1}$) & \small \cite{Neilan}\\
\small $\omega_2$ & \small Efficacy vaccination waning rate 
& \small 1/1 460 (day$^{-1}$) & \small \cite{Cholera_vaccine} \\
\small $\varphi$ & \small Vaccination rate 
& \small 5/1 000 (day$^{-1}$) & \small Assumed \\
\small $\delta$ & \small Treatment rate 
& \small 1.15 (day$^{-1}$) & \small Assumed\\
\small $\varepsilon$ & \small Recovery rate 
& \small 0.2 (day$^{-1}$) & \small \cite{Mwasa}\\
\small $\alpha_1$ & \small Death rate (infected) 
& \small 6$\times10^{-6}$ (day$^{-1}$) & \small \cite{DeathRate,WHO}\\
\small $\alpha_2$ & \small Death rate (in treatment)
& \small 3$\times10^{-6}$ (day$^{-1}$) & \small Assumed\\
\small $\eta$ & \small Shedding rate (infected) 
& \small 10 (cell/ml day$^{-1}$ person$^{-1}$) & \small \cite{Capone}\\
\small $d$ & \small Bacteria death rate 
& \small0.33 (day$^{-1}$) & \small \cite{Capone}\\
\small $S(0)$ & \small Susceptible individuals at $t=0$ 
& \small 28 249 670 (person) & \small \cite{Yemen_pop} \\
\small $I(0)$ & \small Infected individuals at $t=0$ 
& \small 750 (person) & \small \cite{WHO}\\
\small $T(0)$ & \small Treated individuals at $t=0$ 
& \small 0 (person) & \small Assumed \\
\small $R(0)$ & \small Recovered individuals at $t=0$ 
& \small 0 (person) & \small Assumed \\
\small $V(0)$ & \small Vaccinated individuals at $t=0$ 
& \small 0 (person) & \small \cite{WHO_vaccination} \\
\small $B(0)$ & \small Bacterial concentration at $t=0$ 
& \small$275\times10^3$ (cell/ml) & \small Assumed\\
\hline\hline\hline
\end{tabular*}
}
\end{table}


\section{Conclusion}
\label{sec:conclusion}

In this paper, we proposed and analysed, analytically and numerically, 
a SITRVB model for cholera transmission dynamics. In order to fit 
the biggest cholera outbreak worldwide, which has occurred 
very recently in Yemen, we simulated the outbreak 
of Yemen without vaccination. Indeed, vaccination did not exist 
in Yemen from 27$^{\text{th}}$ April 2017 to 15$^{\text{th}}$ 
April 2018. Simulations of our mathematical model, 
with and without vaccination, show that the introduction of vaccination  
from the beginning of the epidemic could have changed the situation 
in Yemen substantially, to the case $R_0<1$, 
where the disease extinguishes naturally. 
Therefore, our research motivates and fortify the 
importance of vaccination in cholera epidemics.


\section*{Acknowledgments}

This research was supported by the
Portuguese Foundation for Science and Technology (FCT)
within projects UID/MAT/04106/2013 (CIDMA), 
and PTDC/EEI-AUT/2933/2014 (TOCCATA), funded by Project 
3599 -- Promover a Produ\c{c}\~ao Cient\'{\i}fica e Desenvolvimento
Tecnol\'ogico e a Constitui\c{c}\~ao de Redes Tem\'aticas, 
and FEDER funds through COMPETE 2020, Programa Operacional
Competitividade e Internacionaliza\c{c}\~ao (POCI).
Lemos-Pai\~{a}o is also supported by the FCT Ph.D.
fellowship PD/BD/114184/2016.

The authors are very grateful to an anonymous referee for reading their paper
carefully and for several constructive remarks, questions and suggestions.


\section*{Conflict of Interest}

The authors declare that there is no conflicts of interest in this paper.



\end{document}